\documentclass[12pt]{article}%
\usepackage{amsmath}
\usepackage{amsfonts}
\usepackage{amssymb}
\usepackage{graphicx}%
\newtheorem{theorem}{Theorem}
\newtheorem{acknowledgement}[theorem]{Acknowledgement}

\newtheorem{definition}[theorem]{Definition}

\newtheorem{proposition}[theorem]{Proposition}
\newtheorem{remark}[theorem]{Remark}

\newenvironment{proof}[1][Proof]{\noindent\textbf{#1.} }{\ \rule{0.5em}{0.5em}}
\begin{document}

\title{Biquaternions for analytic and numerical solution of equations of electrodynamics}
\author{Kira V. Khmelnytskaya, Vladislav V. Kravchenko\\Department of Mathematics\\CINVESTAV del IPN, Queretaro\\Libramiento Norponiente No. 2000\\Fracc. Real de Juriquilla\\Queretaro, Qro.\\C.P. 76230\\MEXICO\\e-mail: vkravchenko@qro.cinvestav.mx}
\date{12/02/2007}
\maketitle

\begin{abstract}
We give an overview of recent advances in analysis of equations of
electrodynamics with the aid of biquaternionic technique. We discuss both
models with constant and variable coefficients, integral representations of
solutions, a numerical method based on biquaternionic fundamental solutions
for solving standard electromagnetic scattering problems, relations between
different operators of mathematical physics including the Schr\"{o}dinger, the
Maxwell system, the conductivity equation and others leading to a deeper
understanding of physics and mathematical properties of the equations.

\end{abstract}

\section{Introduction}

Application of the algebra of biquaternions to equations of electromagnetism
has been subject of an important number of research articles and books (see,
e.g., \cite{Baylis}, \cite{BKT}, \cite{GS2}, \cite{GurseyTze}, \cite{Imaeda},
\cite{AQA}, \cite{KSbook}, \cite{LanczosPhD}, \cite{MM}, \cite{SprossigCubo},
\cite{SprossigFinland}, \cite{Tanisli}\ and many others). The aim of this work
is to present some recent results in the field concerning the usage of
algebraic advantages of biquaternions for analytic and numerical solution of
Maxwell's system for chiral media as well as for inhomogeneous media. Compared
to a considerable number of publications dealing with biquaternionic
reformulations of Maxwell's equations for a vacuum or for a homogeneous
isotropic medium, application of biquaternions to electromagnetic models
corresponding to more complicated media (a much more challenging object for
studying) was discussed in relatively few sources (\cite{GKK}, \cite{KKO},
\cite{KKR}, \cite{AQA}, \cite{KOZAA},\ \cite{SprossigCubo},
\cite{SprossigFinland}). Meanwhile the possibility of representation of
Maxwell's system for a vacuum in the form of a single biquaternionic equation
is known since 1919 \cite{LanczosPhD}, only recently it became clear how this
result can be generalized for inhomogeneous media \cite{KrISAAC}, \cite{AQA}
and for chiral media \cite{GKK}. An appropriate quaternionic or biquaternionic
reformulation of a first order system of mathematical physics opens the way
for applying different methods which in many aspects preserve the algebraic
power of complex analysis. For example, it is not easy to arrive at the Cauchy
integral formula for holomorphic functions using two-component vector
formalism or even more difficult task to develop using this formalism a
holomorphic function into a Taylor series. No mathematician \ would consider
such way of presenting complex function theory helpful or appropriate.
However, this is precisely what is happening in the study of three or
four-dimensional models of mathematical physics. Compare, e.g., the
Stratton-Chu integrals written in their standard form (see, e.g., \cite{CK1})
with their biquaternionic representation \cite{KrDoklady}, \cite{AQA},
\cite{KSbook} which is in fact a convolution of a biquaternionic fundamental
solution of the Maxwell operator with the electromagnetic field and it is
quite evident that the latter is natural and elucidating. The meaning of the
Stratton-Chu formulas as a Cauchy integral formula for the electromagnetic
field becomes transparent and no doubt in their biquaternionic form the
Stratton-Chu formulas could be included in a moderately advanced course of
electromagnetic theory which is not the usual case up to now in spite of their
central role in electrical engineering applications.

The results presented in this work are \textquotedblleft essentially
quaternionic\textquotedblright\ in the sense that it is not clear how they
could be obtained with other techniques. In the first part (sections 2-5) we
explain a numerical method for solving electromagnetic scattering problems
with an unusual for three-dimensional models precision by the aid of
biquaternionic fundamental solutions which to the difference of the usually
utilized matrix fundamental solutions for the Maxwell equations (see, e.g.,
\cite{Alexbook} and \cite{Eremin}) enjoy some advantageous properties. First
of all they are not matrices but vectors (of four components). Second, they
have a clear physical meaning of fields generated by point sources. Third,
their singularity is lower than that of fundamental solutions based on a
matrix approach. The main idea of this work is to explain how our approach
works and how it can be used. We only formulate some necessary results like
those about the completeness of our systems of quaternionic fundamental
solutions in appropriate functional spaces referring the interested reader to
some previous publications, in particular \cite{KKR} where the corresponding
proofs can be found.

In sections 6-8 we consider the time-dependent Maxwell system for chiral
media, rewrite it in a biquaternionic form as a single equation and then
construct explicitly a corresponding Green function. Section 9 is dedicated to
the time-dependent Maxwell equations for inhomogeneous media. We show that
these equations can also be written as a single biquaternionic equation. In a
static case the corresponding quaternionic operator factorizes the stationary
Schr\"{o}dinger operator. We study relationship between solutions of these
important physical equations.

\section{Biquaternionic fundamental solutions}

Let $\mathbb{H}(\mathbb{C})$ denote the set of complex quaternions (=
biquaternions). Each element $a$ of $\mathbb{H}(\mathbb{C})$ is represented in
the form $a=\sum_{k=0}^{3}a_{k}i_{k}$ where $\{a_{k}\}\subset\mathbb{C}$,
$i_{0}$ is the unit and $\{i_{k}|\quad k=1,2,3\}$ are the quaternionic
imaginary units:%

\[
i_{0}^{2}=i_{0}=-i_{k}^{2};\;i_{0}i_{k}=i_{k}i_{0}=i_{k},\quad k=1,2,3;
\]

\[
i_{1}i_{2}=-i_{2}i_{1}=i_{3};\;i_{2}i_{3}=-i_{3}i_{2}=i_{1};\;i_{3}%
i_{1}=-i_{1}i_{3}=i_{2}.
\]
We denote the imaginary unit in $\mathbb{C}$ by $i$ as usual. By
definition\ \ $i$\ commutes with \ $i_{k}$, $k=\overline{0,3}$.

We will use the vector representation of complex quaternions, every
$a\in\mathbb{H}(\mathbb{C})$ is represented as follows $a=a_{0}%
+\overrightarrow{a}$, where $a_{0}$ is the scalar part of $a$:
$\operatorname*{Sc}(a)=a_{0}$, and $\overrightarrow{a}$ is the vector part of
$a$: $\operatorname*{Vec}(a)=\overrightarrow{a}$ $=\sum_{k=1}^{3}a_{k}i_{k}$.
Complex quaternions of the form $a=\overrightarrow{a}$ are called purely
vectorial and can be identified with vectors from $\mathbb{C}^{3}$. The
operator of quaternionic conjugation we denote by $C_{H}$: $\overline{a}%
=C_{H}a=a_{0}-\overrightarrow{a}$.

Let us introduce the operator $D=\sum_{k=1}^{3}i_{k}\partial_{k}$, where
$\partial_{k}=\frac{\partial}{\partial x_{k}}$, whose action on quaternion
valued functions can be represented in a vector form as follows%

\[
Df=-\operatorname*{div}\overrightarrow{f}+\operatorname*{grad}f_{0}%
+\operatorname*{rot}\overrightarrow{f}.
\]
That is, $\operatorname*{Sc}(Df)=-\operatorname*{div}\overrightarrow{f}$ and
$\operatorname*{Vec}(Df)=\operatorname*{grad}f_{0}+\operatorname*{rot}%
\overrightarrow{f}$.

Denote $D_{\alpha}=D+\alpha$, where $\alpha$ is a complex constant. We have
the following factorization of the Helmholtz operator \cite{GuZAA}:
\begin{equation}
\Delta+\alpha^{2}=-D_{\alpha}D_{-\alpha}=-D_{-\alpha}D_{\alpha}.
\label{factor}%
\end{equation}
Using the fundamental solution of the Helmholtz operator%
\[
\theta_{\alpha}(x)=-\frac{e^{i\alpha\left\vert x\right\vert }}{4\pi\left\vert
x\right\vert }%
\]
(we suppose that $\operatorname*{Im}\alpha\geq0$), the fundamental solutions
$\mathcal{K}_{\alpha}$ and $\mathcal{K}_{-\alpha}$ for the operators
$D_{\alpha}$ and $D_{-\alpha}$ can be obtained from (\ref{factor}) in the
following way%

\begin{equation}
\mathcal{K}_{\alpha}=-(D-\alpha)\theta_{\alpha}\qquad\text{and\qquad
}\mathcal{K}_{-\alpha}=-(D+\alpha)\theta_{\alpha}. \label{fund}%
\end{equation}
We have%
\[
D_{\pm\alpha}\mathcal{K}_{\pm\alpha}=\delta,
\]
where $\delta$ is the Dirac delta function.

From (\ref{fund}) we obtain the explicit form of $\mathcal{K}_{\alpha}$ and
$\mathcal{K}_{-\alpha}$:%
\begin{equation}
\mathcal{K}_{\pm\alpha}(x)=(\pm\alpha+\frac{x}{\left\vert x\right\vert ^{2}%
}-i\alpha\frac{x}{\left\vert x\right\vert })\theta_{\alpha}(x). \label{fund_d}%
\end{equation}
Here $x=\sum_{k=1}^{3}x_{k}i_{k}$. Note that $\mathcal{K}_{\alpha}$ and
$\mathcal{K}_{-\alpha}$ are full biquaternions with $\operatorname*{Sc}%
(\mathcal{K}_{\pm\alpha}(x))=\pm\alpha\theta_{\alpha}(x)$ and
$\operatorname*{Vec}(\mathcal{K}_{\pm\alpha}(x))=-\operatorname*{grad}%
\theta_{\alpha}(x)=(\frac{x}{\left\vert x\right\vert ^{2}}-i\alpha\frac
{x}{\left\vert x\right\vert })\theta_{\alpha}(x)$.

More information on the algebra of biquaternions and related calculus can be
found in \cite{AQA}.

\section{Biquaternionic reformulation of Maxwell's equations in chiral media}

The operators $D_{\alpha}$ and $D_{-\alpha}$ are closely related to the
Maxwell equations. Consider the Maxwell system for a homogeneous chiral medium
(see, e.g., \cite{LVV, Lindell})%

\begin{equation}
\operatorname{rot}\overrightarrow{E}\left(  x\right)  =-i\alpha\left(
\overrightarrow{H}\left(  x\right)  +\beta\operatorname{rot}\overrightarrow
{H}\left(  x\right)  \right)  \label{M12}%
\end{equation}
and%
\begin{equation}
\operatorname{rot}\overrightarrow{H}\left(  x\right)  =i\alpha\left(
\overrightarrow{E}\left(  x\right)  +\beta\operatorname{rot}\overrightarrow
{E}\left(  x\right)  \right)  , \label{M13}%
\end{equation}
where $\alpha=\omega\sqrt{\varepsilon\mu}$. Some examples of numerical values
of $\beta$ for physical media can be found, e.g., in \cite{Lindell}. We notice
only that when $\beta=0$ we obtain the Maxwell system for a homogeneous,
isotropic achiral medium with the wave number $\alpha$.

The vectors $\overrightarrow{E}$ and $\overrightarrow{H}$ in (\ref{M12}) and
(\ref{M13}) are complex. Consider the following purely vectorial
biquaternionic functions%
\[
\overrightarrow{\varphi}=\overrightarrow{E}+i\overrightarrow{H}\qquad
\text{and\qquad}\overrightarrow{\psi}=\overrightarrow{E}-i\overrightarrow{H}.
\]
It is easy to verify (see \cite{KKO}, \cite{KOZAA} or \cite{KKR}) that
$\overrightarrow{\varphi}$ and $\overrightarrow{\psi}$ satisfy the following equations%

\[
\left(  D+\alpha_{1}\right)  \overrightarrow{\varphi}=0
\]
and
\[
\left(  D-\alpha_{2}\right)  \overrightarrow{\psi}=0,
\]
where
\[
\alpha_{1}=\frac{\alpha}{(1+\alpha\beta)},\qquad\alpha_{2}=\frac{\alpha
}{(1-\alpha\beta)}.
\]

\begin{remark}
If $\beta=0$ then $\alpha_{1}=\alpha_{2}=\alpha$ and we arrive at the
quaternionic form of the Maxwell equations in the achiral case (see
\cite[Sect. 9]{KSbook}, \cite{AQA}) but in general $\alpha_{1}$ and
$\alpha_{2}$ are different and physically characterize the propagation of
electromagnetic waves of opposing circular polarizations.
\end{remark}

Obviously the vectors $\overrightarrow{E}$ and $\overrightarrow{H}$ are easily
recovered from $\overrightarrow{\varphi}$ and $\overrightarrow{\psi}$:%
\[
\overrightarrow{E}=\frac{1}{2}(\overrightarrow{\varphi}+\overrightarrow{\psi
})\qquad\text{\ and \qquad}\overrightarrow{H}=\frac{1}{2i}(\overrightarrow
{\varphi}-\overrightarrow{\psi}).
\]

\section{Completeness of a system of biquaternionic fundamental solutions}

Let $\Gamma$ be a sufficiently smooth closed surface in $\mathbb{R}^{3}$. Here
we use the term sufficiently smooth for surfaces whose smoothness allows us to
introduce the corresponding Sobolev space $H^{s}(\Gamma)$ for a given
$s\in\mathbb{R}$.

The interior domain enclosed by $\Gamma$ we denote by $\Omega^{+}$ and the
exterior by $\Omega^{-}$.

Let $\overrightarrow{e}$ and $\overrightarrow{h}$ be two complex vectors
defined on $\Gamma$.

\begin{definition}
We say that $\overrightarrow{e}$ and $\overrightarrow{h}$ are extendable into
$\Omega^{+}$ if there exist such pair of vectors $\overrightarrow{E}$ and
$\overrightarrow{H}$ defined on $\overline{\Omega^{+}}$ that equations
(\ref{M12}) and (\ref{M13}) are satisfied in $\Omega^{+}$ and on $\Gamma$ we
have $\overrightarrow{E}\mid_{\Gamma}=\overrightarrow{e}$ and $\overrightarrow
{H}\mid_{\Gamma}=\overrightarrow{h}$.
\end{definition}

\begin{definition}
The vectors $\overrightarrow{e}$ and $\overrightarrow{h}$ are extendable into
$\Omega^{-}$ if there exist such pair of vectors $\overrightarrow{E}$ and
$\overrightarrow{H}$ defined on $\overline{\Omega^{-}}$ that equations
(\ref{M12}) and (\ref{M13}) are satisfied in $\Omega^{-}$, the
Silver-M\"{u}ller condition
\begin{equation}
\overrightarrow{E}-\left[  \frac{x}{\left|  x\right|  }\times\overrightarrow
{H}\right]  =o(\frac{1}{\left|  x\right|  }) \label{SM}%
\end{equation}
is fulfilled at infinity uniformly for all directions and on $\Gamma$:
$\overrightarrow{E}\mid_{\Gamma}=\overrightarrow{e}$ and $\overrightarrow
{H}\mid_{\Gamma}=\overrightarrow{h}$.
\end{definition}

With the aid of quaternionic analysis techniques these two introduced classes
of vector functions can be completely described. For achiral media it was done
in \cite{Krdep} (see also \cite[Sect. 11]{KSbook}) and for chiral media in
\cite{KKO}. Here we recall these results without proof.

We will need the following operators%
\begin{equation}
S_{\alpha}f(x)=-2\int_{\Gamma}\mathcal{K}_{\alpha}(x-y)\overrightarrow
{n}(y)f(y)d\Gamma_{y},\quad x\in\Gamma, \label{Sa}%
\end{equation}%
\[
P_{\alpha}=\frac{1}{2}(I+S_{\alpha})\quad\text{and\quad}Q_{\alpha}=\frac{1}%
{2}(I-S_{\alpha})
\]
which are bounded in $H^{s}(\Gamma)$ for all real $s$. The function $f$ in
(\ref{Sa}) is a biquaternion valued function, $\overrightarrow{n}$ is the
quaternionic representation of the outward with respect to $\Omega^{+}$
unitary normal on $\Gamma$: $\overrightarrow{n}=\sum\nolimits_{k=1}^{3}%
n_{k}i_{k}$ and all the products in the integrand in (\ref{Sa}) are
quaternionic products. From the numerous interesting properties of the
operators $P_{\alpha}$, $Q_{\alpha}$ and $S_{\alpha}$ (see \cite{KSbook}) we
will need only the following fact

\begin{theorem}
Let complex vectors $\overrightarrow{e}$ and $\overrightarrow{h}$ belong to
$H^{s}(\Gamma)$, $s>0$. Then

\begin{enumerate}
\item in order for $\overrightarrow{e}$ and $\overrightarrow{h}$ to be
extendable into $\Omega^{+}$ the following condition is necessary and
sufficient%
\begin{equation}
(\overrightarrow{e}+i\overrightarrow{h})\in\operatorname{im}P_{\alpha_{1}%
}(H^{s}(\Gamma))\text{\qquad and\qquad}(\overrightarrow{e}-i\overrightarrow
{h})\in\operatorname{im}P_{-\alpha_{2}}(H^{s}(\Gamma)) \label{iffint}%
\end{equation}
or which is the same
\[
\overrightarrow{e}+i\overrightarrow{h}=S_{\alpha_{1}}(\overrightarrow
{e}+i\overrightarrow{h})\text{\qquad and\qquad}\overrightarrow{e}%
-i\overrightarrow{h}=S_{-\alpha_{2}}(\overrightarrow{e}-i\overrightarrow
{h})\qquad\text{on }\Gamma\text{.}%
\]

\item in order for $\overrightarrow{e}$ and $\overrightarrow{h}$ to be
extendable into $\Omega^{-}$ the following condition is necessary and
sufficient%
\begin{equation}
(\overrightarrow{e}+i\overrightarrow{h})\in\operatorname{im}Q_{\alpha_{1}%
}(H^{s}(\Gamma))\text{\qquad and\qquad}(\overrightarrow{e}-i\overrightarrow
{h})\in\operatorname{im}Q_{-\alpha_{2}}(H^{s}(\Gamma)) \label{iffext}%
\end{equation}
or which is the same
\[
\overrightarrow{e}+i\overrightarrow{h}=-S_{\alpha_{1}}(\overrightarrow
{e}+i\overrightarrow{h})\text{\qquad and\qquad}\overrightarrow{e}%
-i\overrightarrow{h}=-S_{-\alpha_{2}}(\overrightarrow{e}-i\overrightarrow
{h})\qquad\text{on }\Gamma\text{.}%
\]

\end{enumerate}
\end{theorem}

Now we will show how two systems of quaternionic fundamental solutions
suitable for the approximation of the vector functions extendable into
$\Omega^{+}$ or $\Omega^{-}$ can be constructed.

By $\Gamma^{-}$ we denote a closed surface enclosed in $\Omega^{+}$ and being
a boundary of a bounded domain $V$, and by $\Gamma^{+}$ we denote a closed
surface enclosing $\overline{\Omega^{+}}$ as shown in the figure.%
%TCIMACRO{\FRAME{ftbpF}{3.736in}{3.0554in}{0pt}{}{}{fig1.jpeg}%
%width 3.736in;  height 3.0554in;  depth 0pt;  original-width 3.54in;
%original-height 2.89in;  cropleft "0";  croptop "1";  cropright "1";
%cropbottom "0";  filename 'fig1.bmp';file-properties "XNPEU";}}}%
%BeginExpansion
\begin{figure}
[ptb]
\begin{center}
\includegraphics[
natheight=2.890000in,
natwidth=3.540000in,
height=3.0554in,
width=3.736in
]%
{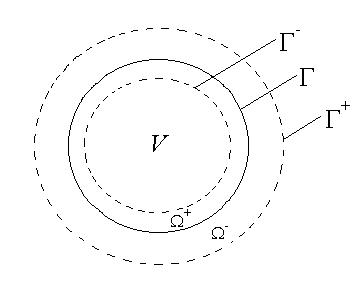}%
\end{center}
\end{figure}
%EndExpansion
By $\left\{  y_{n}^{-}\right\}  _{n=1}^{\infty}$ we denote a set of points
densely distributed on $\Gamma^{-}$, and by $\left\{  y_{n}^{+}\right\}
_{n=1}^{\infty}$ a set of points densely distributed on $\Gamma^{+}$. For each
of these two sets we construct a corresponding pair of systems of quaternionic
fundamental solutions. The pair of systems%
\begin{equation}
\left\{  \mathcal{K}_{\alpha_{1},n}^{+}(x)=\mathcal{K}_{\alpha_{1}}%
(x-y_{n}^{+})\right\}  _{n=1}^{\infty}\text{\qquad and\qquad}\left\{
\mathcal{K}_{-\alpha_{2},n}^{+}(x)=\mathcal{K}_{-\alpha_{2}}(x-y_{n}%
^{+})\right\}  _{n=1}^{\infty}\label{syst+}%
\end{equation}
corresponds to $\left\{  y_{n}^{+}\right\}  _{n=1}^{\infty}$ and the pair of
systems%
\begin{equation}
\left\{  \mathcal{K}_{\alpha_{1},n}^{-}(x)=\mathcal{K}_{\alpha_{1}}%
(x-y_{n}^{-})\right\}  _{n=1}^{\infty}\text{\qquad and\qquad}\left\{
\mathcal{K}_{-\alpha_{2},n}^{-}(x)=\mathcal{K}_{-\alpha_{2}}(x-y_{n}%
^{-})\right\}  _{n=1}^{\infty}\label{syst-}%
\end{equation}
corresponds to $\left\{  y_{n}^{-}\right\}  _{n=1}^{\infty}$.

The following theorems show us the possibility to apply the fundamental
solutions (\ref{syst+}) for the numerical solution of interior boundary value
problems for the Maxwell equations (\ref{M12}), (\ref{M13}), and fundamental
solutions (\ref{syst-}) for the solution of exterior problems.

\begin{theorem}
\label{int}\cite{KKR} Let two complex vectors $\overrightarrow{e}$ and
$\overrightarrow{h}$ belong to $H^{s}(\Gamma)$, $s>1$, be extendable into
$\Omega^{+}$ and both $\alpha_{1}^{2}$ and $\alpha_{2}^{2}$ be not eigenvalues
of the Dirichlet problem in $\Omega^{+}$. Then $\overrightarrow{e}$ and
$\overrightarrow{h}$ can be approximated with an arbitrary precision (in the
norm of $H^{s-1}(\Gamma)$) by right linear combinations of the form%
\[
\overrightarrow{e}_{N}=\frac{1}{2}(\sum_{j=1}^{N}\mathcal{K}_{\alpha_{1}%
,j}^{+}a_{j}+\sum_{j=1}^{N}\mathcal{K}_{-\alpha_{2},j}^{+}b_{j})
\]
and%
\[
\overrightarrow{h}_{N}=\frac{1}{2i}(\sum_{j=1}^{N}\mathcal{K}_{\alpha_{1}%
,j}^{+}a_{j}-\sum_{j=1}^{N}\mathcal{K}_{-\alpha_{2},j}^{+}b_{j}),
\]
where $a_{j}$ and $b_{j}$ are constant complex quaternions.
\end{theorem}

\begin{theorem}
\label{ext}\cite{KKR} Let two complex vectors $\overrightarrow{e}$ and
$\overrightarrow{h}$ belong to $H^{s}(\Gamma)$, $s>1$, be extendable into
$\Omega^{-}$ and let both $\alpha_{1}^{2}$ and $\alpha_{2}^{2}$ be not
eigenvalues of the Dirichlet problem in $V$. Then $\overrightarrow{e}$ and
$\overrightarrow{h}$ can be approximated with an arbitrary precision (in the
norm of $H^{s-1}(\Gamma)$) by right linear combinations of the form%
\begin{equation}
\overrightarrow{e}_{N}=\frac{1}{2}(\sum_{j=1}^{N}\mathcal{K}_{\alpha_{1}%
,j}^{-}a_{j}+\sum_{j=1}^{N}\mathcal{K}_{-\alpha_{2},j}^{-}b_{j}) \label{eext}%
\end{equation}
and%
\begin{equation}
\overrightarrow{h}_{N}=\frac{1}{2i}(\sum_{j=1}^{N}\mathcal{K}_{\alpha_{1}%
,j}^{-}a_{j}-\sum_{j=1}^{N}\mathcal{K}_{-\alpha_{2},j}^{-}b_{j}), \label{hext}%
\end{equation}
where $a_{j}$ and $b_{j}$ are constant complex quaternions.
\end{theorem}

\begin{remark}
The right linear combinations in Theorem \ref{int} and Theorem \ref{ext} are
in general full quaternions. In order to ensure that they will be purely
vectorial additionally to a usual boundary condition for the electromagnetic
field we have to add the requirement that their scalar parts be equal to zero.
We show how this can be easily achieved on some examples of numerical
realization considered in the next section.
\end{remark}

\section{Numerical realization}

Consider the exterior boundary value problem for the Maxwell equations
corresponding to the model of electromagnetic scattering by a perfectly
conducting body with a boundary $\Gamma$. Find two vectors $\overrightarrow
{E}$ and $\overrightarrow{H}$ satisfying (\ref{M12}) and (\ref{M13}) in
$\Omega^{-}$, the condition (\ref{SM}) at infinity and the following boundary
condition%
\begin{equation}
\left[  \overrightarrow{E}(x)\times\overrightarrow{n}(x)\right]
=\overrightarrow{f}(x),\qquad x\in\Gamma, \label{b1}%
\end{equation}
where $\overrightarrow{f}$ is a given tangential field.

We look for the solutions in the form (\ref{eext}) and (\ref{hext}), applying
the collocation method in order to find the coefficients $a_{j}$ and $b_{j}$.
Substitution of the vector part of (\ref{eext}) in (\ref{b1}) gives us two
linearly independent equations in every collocation point. In each collocation
point $x\in\Gamma$ we must require also that%
\begin{equation}
\operatorname*{Sc}(\sum_{j=1}^{N}\mathcal{K}_{\alpha_{1},j}^{-}(x)a_{j}%
+\sum_{j=1}^{N}\mathcal{K}_{-\alpha_{2},j}^{-}(x)b_{j})=0, \label{Sc1}%
\end{equation}
and%
\begin{equation}
\operatorname*{Sc}(\sum_{j=1}^{N}\mathcal{K}_{\alpha_{1},j}^{-}(x)a_{j}%
-\sum_{j=1}^{N}\mathcal{K}_{-\alpha_{2},j}^{-}(x)b_{j})=0, \label{Sc2}%
\end{equation}
which gives us other two linearly independent equations. Taking into account
that in (\ref{eext}) and (\ref{hext}) we have $8N$ unknown complex quantities
we need $2N$ collocation points. After having solved the corresponding system
of linear algebraic equations we obtain the coefficients $a_{j}$ and $b_{j}$
and consequently the approximate solution of the problem. A good approximation
of the boundary condition (\ref{b1}) guarantees a good approximation of the
electromagnetic field in the domain $\Omega^{-}$ due to the following estimate
(see \cite[p. 126]{Eremin})%
\[
\left\|  \overrightarrow{E}\right\|  _{\infty,\Omega_{r}^{-}}+\left\|
\overrightarrow{H}\right\|  _{\infty,\Omega_{r}^{-}}\leq C\left\|
\overrightarrow{n}\times\overrightarrow{E}\right\|  _{L_{2}(\Gamma)}%
\]
where $\left\|  \cdot\right\|  _{\infty,\Omega_{r}^{-}}$ stands for the
supremum norm in any closed subset $\Omega_{r}^{-}$ of $\Omega^{-}$ and $C$ is
a positive constant depending on $\Gamma$ and $\Omega_{r}^{-}$.

The method was tested \cite{KiraPhD}, \cite{KKRAlicante}, \cite{KKR} using
different exact solutions. For example, let $\beta=0$ and consequently
$\alpha=\alpha_{1}=\alpha_{2}$. The vectors%

\[
\overrightarrow{E}^{m}(x)=\operatorname*{rot}\overrightarrow{c}\theta_{\alpha
}(x)
\]
and%
\[
\overrightarrow{H}^{m}(x)=-\frac{1}{i\alpha}\operatorname*{rot}\overrightarrow
{E}^{m}(x),\qquad x\in\mathbb{R}^{3}\setminus\left\{  0\right\}  ,
\]
where $\overrightarrow{c}\in\mathbb{R}^{3}$ is constant, represent the
electromagnetic field of a magnetic dipole situated at the origin \cite[Sect.
4.2]{CK1}. They satisfy (\ref{M12}) and (\ref{M13}) (for $\beta=0)$ as well as
the Silver-M\"{u}ller conditions at infinity.

Let $\Gamma$ be an ellipsoid described by the equalities.
\begin{equation}
x_{1}=a\cos\eta\sin\nu,\quad x_{2}=b\sin\eta\sin\nu,\quad x_{3}=c\cos\nu,
\label{ellips}%
\end{equation}
where $0<\eta\leq2\pi$, $0<\nu\leq\not \pi $. Then $\overrightarrow{E}^{m}$
and $\overrightarrow{H}^{m}$ give us the solution of the following boundary
value problem%
\[
\operatorname*{rot}\overrightarrow{E}(x)=-i\alpha\overrightarrow{H}(x),\qquad
x\in\Omega^{-},
\]%
\[
\operatorname*{rot}\overrightarrow{H}(x)=i\alpha\overrightarrow{E}(x),\qquad
x\in\Omega^{-},
\]

\[
\left[  \overrightarrow{E}(x)\times\overrightarrow{n}(x)\right]
=\overrightarrow{f}(x),\qquad x\in\Gamma
\]
where
\[
\overrightarrow{f}(x)=\left[  \left(
\begin{array}
[c]{c}%
c_{3}\partial_{2}\theta_{\alpha}(x)-c_{2}\partial_{3}\theta_{\alpha}(x)\\
c_{1}\partial_{3}\theta_{\alpha}(x)-c_{3}\partial_{1}\theta_{\alpha}(x)\\
c_{2}\partial_{1}\theta_{\alpha}(x)-c_{1}\partial_{2}\theta_{\alpha}(x)
\end{array}
\right)  \times\overrightarrow{n}(x)\right]  .
\]

We give the numerical results for $a=5$, $b=3$ and $c=2$ in (\ref{ellips}). As
the auxiliary surface $\Gamma^{-}$ containing points $y_{n}^{-}$ we have
chosen an ellipsoid interior with respect to $\Gamma$ with $a$, $b$ and $c$
multiplied by $0.15$. In the following table we present the results for
$\alpha=1+0.3i$ and for different values of $N$. The corresponding errors
represent the absolute maximum difference between the exact and the
approximate solutions at the points on the ellipsoid exterior with respect to
$\Gamma$ with $a$, $b$ and $c$ multiplied by $5$.

\begin{center}%
\begin{tabular}
[c]{||c|c|c||}\hline\hline
$N$ & Error for $\overrightarrow{E}$ & Error for $\overrightarrow{H}$\\\hline
10 & 0.441E-03 \  & 0.332E-03 \ \\\hline
15 & 0.693E-05 & 0.713E-05\\\hline
20 & 0.162E-05 & 0.186E-05\\\hline
25 & 0.245E-06 & 0.248E-06\\\hline
30 & 0.113E-06 & 0.171E-06\\\hline
35 & 0.522E-07 & 0.409E-07\\\hline\hline
\end{tabular}

\end{center}

A quite fast convergence of the method can be appreciated (all numerical
results were obtained on a PC Pentium 4).

Let us notice that the approximation by linear combinations of quaternionic
fundamental solutions can be applied to other classes of boundary value
problems for the Maxwell system like for example the impedance problem with
the boundary condition%
\[
\left[  \overrightarrow{E}(x)\times\overrightarrow{n}(x)\right]  -\xi\left[
\left[  \overrightarrow{H}(x)\times\overrightarrow{n}(x)\right]
\times\overrightarrow{n}(x)\right]  =\overrightarrow{f}(x),\qquad x\in\Gamma.
\]
This implies some obvious changes in the matrix of coefficients of the system
of linear algebraic equations corresponding to collocation points.

More results and analysis of numerical experiments were given in
\cite{KiraPhD}.

\section{Time-dependent Maxwell's equations for chiral media\label{Maxw}}

Consider time-dependent Maxwell's equations%

\begin{equation}
\operatorname*{rot}\overrightarrow{E}(t,x)=-\partial_{t}\overrightarrow
{B}(t,x), \label{rot1}%
\end{equation}

\begin{equation}
\operatorname*{rot}\overrightarrow{H}(t,x)=\partial_{t}\overrightarrow
{D}(t,x)+\overrightarrow{j}(t,x), \label{rot2}%
\end{equation}

\begin{equation}
\operatorname*{div}\overrightarrow{E}(t,x)=\frac{\rho(t,x)}{\varepsilon
},\qquad\operatorname*{div}\overrightarrow{H}(t,x)=0 \label{div}%
\end{equation}
with the Drude-Born-Fedorov constitutive relations corresponding to the chiral
media (see, e.g., \cite{ARS}, \cite{Lak}, \cite{Lindell}):
\begin{equation}
\overrightarrow{B}(t,x)=\mu(\overrightarrow{H}(t,x)+\beta\operatorname*{rot}%
\overrightarrow{H}(t,x)), \label{DBF1}%
\end{equation}%
\begin{equation}
\overrightarrow{D}(t,x)=\varepsilon(\overrightarrow{E}(t,x)+\beta
\operatorname*{rot}\overrightarrow{E}(t,x)), \label{DBF2}%
\end{equation}
where $\beta$ is the chirality measure of the medium. $\beta,\varepsilon,\mu$
are real scalars assumed to be constants. Note that the charge density $\rho$
and the current density $\overrightarrow{j}$ are related by the continuity
equation $\partial_{t}\rho+\operatorname*{div}\overrightarrow{j}=0$.

Incorporating the constitutive relations (\ref{DBF1}), (\ref{DBF2}) into the
system (\ref{rot1})-(\ref{div}) we arrive at the time-dependent Maxwell system
for a homogeneous chiral medium%

\begin{equation}
\operatorname*{rot}\overrightarrow{H}(t,x)=\varepsilon(\partial_{t}%
\overrightarrow{E}(t,x)+\beta\partial_{t}\operatorname*{rot}\overrightarrow
{E}(t,x))+\overrightarrow{j}(t,x), \label{Max1}%
\end{equation}%
\begin{equation}
\operatorname*{rot}\overrightarrow{E}(t,x)=-\mu(\partial_{t}\overrightarrow
{H}(t,x)+\beta\partial_{t}\operatorname*{rot}\overrightarrow{H}(t,x)),
\label{Max2}%
\end{equation}

\begin{equation}
\operatorname*{div}\overrightarrow{E}(t,x)=\frac{\rho(t,x)}{\varepsilon
},\qquad\operatorname*{div}\overrightarrow{H}(t,x)=0. \label{Max3}%
\end{equation}

Application of $\operatorname*{rot}$ to (\ref{Max1}) and (\ref{Max2}) allows
us to separate the equations for $\overrightarrow{E}$ and $\overrightarrow{H}$
and to obtain in this way the wave equations for a chiral medium%
\begin{equation}
\operatorname*{rot}\operatorname*{rot}\overrightarrow{E}+\varepsilon
\mu\partial_{t}^{2}\overrightarrow{E}+2\beta\varepsilon\mu\partial_{t}%
^{2}\operatorname*{rot}\overrightarrow{E}+\beta^{2}\varepsilon\mu\partial
_{t}^{2}\operatorname*{rot}\operatorname*{rot}\overrightarrow{E}=-\mu
\partial_{t}\overrightarrow{j}-\beta\mu\partial_{t}\operatorname*{rot}%
\overrightarrow{j}, \label{wave1}%
\end{equation}%
\begin{equation}
\operatorname*{rot}\operatorname*{rot}\overrightarrow{H}+\varepsilon
\mu\partial_{t}^{2}\overrightarrow{H}+2\beta\varepsilon\mu\partial_{t}%
^{2}\operatorname*{rot}\overrightarrow{H}+\beta^{2}\varepsilon\mu\partial
_{t}^{2}\operatorname*{rot}\operatorname*{rot}\overrightarrow{H}%
=\operatorname*{rot}\overrightarrow{j}. \label{wave2}%
\end{equation}

It should be noted that when $\beta=0$, (\ref{wave1}) and (\ref{wave2}) reduce
to the wave equations for non-chiral media but in general to the difference of
the usual non-chiral wave equations their chiral generalizations represent
equations of fourth order.

\section{Field equations in a biquaternionic form}

In this section following \cite{GKK} we rewrite the field equations from
Section \ref{Maxw} in a biquaternionic form.

Let us introduce the following biquaternionic operator
\begin{equation}
M=\beta\sqrt{\varepsilon\mu}\partial_{t}D+\sqrt{\varepsilon\mu}\partial_{t}-iD
\label{A}%
\end{equation}
and consider the purely vectorial biquaternionic function
\begin{equation}
\overrightarrow{V}(t,x)=\overrightarrow{E}(t,x)-i\sqrt{\frac{\mu}{\varepsilon
}}\overrightarrow{H}(t,x). \label{V}%
\end{equation}

\begin{proposition}
\cite{GKK} The equation
\begin{equation}
M\overrightarrow{V}(t,x)=-\sqrt{\frac{\mu}{\varepsilon}}\overrightarrow
{j}(t,x)-\beta\sqrt{\frac{\mu}{\varepsilon}}\partial_{t}\rho(t,x)+\frac
{i\rho(t,x)}{\varepsilon} \label{Amax}%
\end{equation}
is equivalent to the Maxwell system (\ref{Max1})-(\ref{Max3}), the vectors
$\overrightarrow{E}$ and $\overrightarrow{H}$ are solutions of (\ref{Max1}%
)-(\ref{Max3}) if and only if the purely vectorial biquaternionic function
$\overrightarrow{V}$ defined by (\ref{V}) is a solution of (\ref{Amax}).
\end{proposition}

\begin{proof}
The scalar and the vector parts of (\ref{Amax}) have the form%
\begin{equation}
-\beta\sqrt{\varepsilon\mu}\partial_{t}\operatorname*{div}\overrightarrow
{E}+\sqrt{\frac{\mu}{\varepsilon}}\operatorname*{div}\overrightarrow
{H}+i(\operatorname*{div}\overrightarrow{E}+\beta\mu\partial_{t}%
\operatorname*{div}\overrightarrow{H})=-\beta\sqrt{\frac{\mu}{\varepsilon}%
}\partial_{t}\rho+\frac{i\rho}{\varepsilon}, \label{sc}%
\end{equation}%
\begin{equation}
\beta\sqrt{\varepsilon\mu}\partial_{t}\operatorname*{rot}\overrightarrow
{E}+\sqrt{\varepsilon\mu}\partial_{t}\overrightarrow{E}-\sqrt{\frac{\mu
}{\varepsilon}}\operatorname*{rot}\overrightarrow{H}-i(\operatorname*{rot}%
\overrightarrow{E}+\beta\mu\partial_{t}\operatorname*{rot}\overrightarrow
{H}+\mu\partial_{t}\overrightarrow{H})=-\sqrt{\frac{\mu}{\varepsilon}%
}\overrightarrow{j}. \label{vec}%
\end{equation}
The real part of (\ref{vec}) coincides with (\ref{Max1}) and the imaginary
part coincides with (\ref{Max2}). Applying divergence to the equation
(\ref{vec}) and using the continuity equation gives us
\[
\partial_{t}\operatorname*{div}\overrightarrow{H}=0\text{\quad and\quad
}\partial_{t}\operatorname*{div}\overrightarrow{E}=\frac{1}{\varepsilon
}\partial_{t}\rho.
\]
Taking into account these two equalities we obtain from (\ref{sc}) that the
vectors $\overrightarrow{E}$ and $\overrightarrow{H}$ satisfy equations
(\ref{Max3}).
\end{proof}

It should be noted that for $\beta=0$ from (\ref{A}) we obtain the
biquaternionic Maxwell operator for a homogeneous achiral medium for which the
following equality is valid%
\[
\varepsilon\mu\partial_{t}^{2}-\Delta_{x}=(\sqrt{\varepsilon\mu}\partial
_{t}+iD)(\sqrt{\varepsilon\mu}\partial_{t}-iD).
\]
In the case under consideration ($\beta\neq0$) we obtain a similar result. Let
us denote by $M^{\ast}$ the complex conjugate operator of $M$:%
\[
M^{\ast}=\beta\sqrt{\varepsilon\mu}\partial_{t}D+\sqrt{\varepsilon\mu}%
\partial_{t}+iD.
\]
For simplicity we consider now a sourceless situation. In this case the
equations (\ref{wave1}) and (\ref{wave2}) are homogeneous and can be
represented as follows
\[
MM^{\ast}\overrightarrow{U}(t,x)=0,
\]
where $\overrightarrow{U}\ $stands for $\overrightarrow{E}$ or for
$\overrightarrow{H}$.

\section{Green function for the operator $M$}

Here we present a procedure from \cite{GKK} which gives us a Green function
for the operator $M$. Consider the equation%
\[
(\beta\sqrt{\varepsilon\mu}\partial_{t}D+\sqrt{\varepsilon\mu}\partial
_{t}-iD)f(t,x)=\delta(t,x).
\]
Applying the Fourier transform $\mathcal{F}$ with respect to the time-variable
$t$ we obtain%
\[
(\beta\sqrt{\varepsilon\mu}i\omega D+\sqrt{\varepsilon\mu}i\omega
-iD)F(\omega,x)=\delta(x),
\]
where $F(\omega,x)=\mathcal{F}\{f(t,x)\}=\int_{-\infty}^{\infty}%
f(t,x)e^{-i\omega t}dt.$ The last equation can be rewritten as follows%
\[
(D+\alpha)(\beta\sqrt{\varepsilon\mu}\omega-1)iF(\omega,x)=\delta(x),
\]
where $\alpha=\frac{\sqrt{\varepsilon\mu}\omega}{\beta\sqrt{\varepsilon\mu
}\omega-1}.$ The fundamental solution of $D_{\alpha}$ is given by
(\ref{fund_d}), so we have%
\[
(\beta\sqrt{\varepsilon\mu}\omega-1)iF(\omega,x)=\mathcal{K}_{\alpha
}(x)=(\alpha+\frac{x}{\left\vert x\right\vert ^{2}}-i\alpha\frac{x}{\left\vert
x\right\vert })\Theta_{\alpha}(x),
\]
from where%

\[
F(\omega,x)=\left[  \frac{i\sqrt{\varepsilon\mu}\omega}{(\beta\sqrt
{\varepsilon\mu}\omega-1)^{2}}\left(  1-\frac{ix}{\left\vert x\right\vert
}\right)  +\frac{ix}{\left\vert x\right\vert ^{2}}\frac{1}{\beta
\sqrt{\varepsilon\mu}\omega-1}\right]  \frac{e^{i\left\vert x\right\vert
\frac{\sqrt{\varepsilon\mu}\omega}{\beta\sqrt{\varepsilon\mu}\omega-1}}}%
{4\pi\left\vert x\right\vert }.
\]
We write it in a more convenient form
\[
F(\omega,x)=\left(  \frac{1}{\left(  \omega-a\right)  ^{2}}A\left(  x\right)
+\frac{1}{\omega-a}B\left(  x\right)  \right)  E\left(  x\right)
e^{\frac{ic(x)}{\omega-a}},
\]
where $a=\frac{1}{\beta\sqrt{\varepsilon\mu}}$, $c\left(  x\right)
=\frac{\left\vert x\right\vert }{\beta^{2}\sqrt{\varepsilon\mu}}$, $E\left(
x\right)  =\frac{e^{\frac{i\left\vert x\right\vert }{\beta}}}{4\pi\left\vert
x\right\vert },$%
\[
A\left(  x\right)  =\frac{i}{\beta^{3}\varepsilon\mu}\left(  1-\frac
{ix}{\left\vert x\right\vert }\right)  ,\qquad B\left(  x\right)  =\frac
{i}{\beta\sqrt{\varepsilon\mu}}\left(  \frac{1}{\beta}\left(  1-\frac
{ix}{\left\vert x\right\vert }\right)  +\frac{x}{\left\vert x\right\vert ^{2}%
}\right)  .
\]
In order to obtain the fundamental solution $f(t,x)$ we should apply the
inverse Fourier transform to $F(\omega,x)$. Among different regularizations of
the resulting integral we should choose the one leading to a fundamental
solution satisfying the causality principle, that is vanishing for $t<0$. Such
an election is done by introducing a small parameter $y>0$ in the following way%

\begin{equation}
f(t,x)=\lim_{y\rightarrow0}\mathcal{F}^{-1}\left\{  F(z,x)\right\}  \label{f}%
\end{equation}
where $z=\omega-iy$. This regularization is in agreement with the condition
$\operatorname*{Im}\alpha\geq0$. We have%
\begin{equation}
\mathcal{F}^{-1}\left\{  F(z,x)\right\}  =\frac{1}{2\pi}\int_{-\infty}%
^{\infty}\left(  \frac{1}{\left(  \omega-a_{y}\right)  ^{2}}A\left(  x\right)
+\frac{1}{\omega-a_{y}}B\left(  x\right)  \right)  E\left(  x\right)
e^{\frac{ic(x)}{\omega-a_{y}}}e^{i\omega t}d\omega\label{fund1}%
\end{equation}
where $a_{y}=a+iy$. Expression (\ref{fund1}) includes two integrals of the form%

\[
I_{k}=\frac{1}{2\pi}\int_{-\infty}^{\infty}\frac{e^{\frac{ic}{\omega-a_{y}}%
}e^{i\omega t}}{\left(  \omega-a_{y}\right)  ^{k}}d\omega,\quad k=1,2
\]
where $c=c(x)$. We have
\begin{equation}
I_{k}=\frac{1}{2\pi}\sum_{j=0}^{\infty}\left(  \frac{\left(  ic\right)  ^{j}%
}{j!}\int_{-\infty}^{\infty}\frac{e^{i\omega t}d\omega}{\left(  \omega
-a_{y}\right)  ^{j+k}}\right)  , \label{Ik}%
\end{equation}
where the change of order of integration and summation is possible because the
two necessary for this conditions are fulfilled: the series is uniformly
convergent on each segment and the integrals of partial sums converge
uniformly with respect to $j$. Denote
\[
I_{k,j}(t)=\int_{-\infty}^{\infty}\frac{e^{i\omega t}d\omega}{\left(
\omega-a_{y}\right)  ^{j+k}}.
\]
For $k=1$ and $j=0$ we obtain (see, e.g., \cite[Sect. 8.7]{Brem})
\[
I_{1,0}(t)=2\pi iH(t)e^{ita_{y}}%
\]
where $H$ is the Heaviside function. For all other cases, that is for $k=1$
and $j=\overline{1,\infty}$ and for $k=2$ and $j=\overline{0,\infty}$ \ we
have that $j+k\geq2$ and the integrand in (\ref{Ik}) has a pole at the point
$a_{y}$ of order $j+k$. Using a result from the residue theory \cite[Sect.
4.3]{Derrick} we obtain%
\[
I_{k,j}(t)=2\pi i\operatorname*{Res}_{a_{y}}\frac{e^{i\omega t}}{\left(
\omega-a_{y}\right)  ^{j+k}}\quad\text{for }t\geq0\text{ and }j+k\geq2.
\]
Consider%
\[
\operatorname*{Res}_{a_{y}}\frac{e^{i\omega t}}{\left(  \omega-a_{y}\right)
^{j+k}}=\frac{1}{(j+k-1)!}\lim_{\omega\rightarrow a_{y}}\frac{\partial
^{j+k-1}}{\partial\omega^{j+k-1}}e^{i\omega t}=\frac{(it)^{j+k-1}e^{ia_{y}t}%
}{(j+k-1)!}\quad\text{for }t\geq0
\]
and $j+k\geq2.$

For $t<0$ we have that $I_{k,j}(t)$ is equal to the sum of residues with
respect to singularities in the lower half-plane $y<0$ which is zero because
the integrand is analytic there. Thus we obtain%
\[
I_{k,j}(t)=2\pi iH(t)\frac{(it)^{j+k-1}}{(j+k-1)!}e^{ia_{y}t}.
\]
Substitution of this result into (\ref{Ik}) gives us%
\[
I_{1}=iH(t)e^{ia_{y}t}\sum_{j=0}^{\infty}\frac{(-ct)^{j}}{j!j!}\qquad
\text{and}\qquad I_{2}=-H(t)e^{ia_{y}t}t\sum_{j=0}^{\infty}\frac{(-ct)^{j}%
}{j!(j+1)!}.
\]
Now using the series representations of the Bessel functions $J_{0}$ and
$J_{1}$ (see e.g. \cite[Chapter 5]{Vladimirov}) we obtain
\[
I_{1}=iH(t)e^{ia_{y}t}J_{0}\left(  2\sqrt{ct}\right)  \text{\quad and\quad
}I_{2}=-H(t)\sqrt{\frac{t}{c}}e^{ia_{y}t}J_{1}\left(  2\sqrt{ct}\right)  .
\]
Substituting these expressions in (\ref{fund1}) and then in (\ref{f}) we
arrive at the following expression for $f$:%

\[
f(t,x)=H(t)e^{iat}E\left(  x\right)  \left(  -A\left(  x\right)  \sqrt
{\frac{t}{c}}J_{1}\left(  2\sqrt{ct}\right)  +iB\left(  x\right)  J_{0}\left(
2\sqrt{ct}\right)  \right)  .
\]
Finally we rewrite the obtained fundamental solution of the operator $M$ in
explicit form:%
\begin{align*}
f(t,x)  &  =H(t)\frac{e^{\frac{it}{\beta\sqrt{\varepsilon\mu}}}}{\beta
\sqrt{\varepsilon\mu}}\left(  \mathcal{K}_{\frac{1}{\beta}}(x)J_{0}\left(
\frac{2\sqrt{t\left|  x\right|  }}{\beta\left(  \varepsilon\mu\right)
^{\frac{1}{4}}}\right)  \right. \\
&  \left.  +\frac{i\Theta_{\frac{1}{\beta}}(x)}{\beta\left(  \varepsilon
\mu\right)  ^{\frac{1}{4}}}\left(  1-\frac{ix}{\left|  x\right|  }\right)
\sqrt{\frac{t}{\left|  x\right|  }}J_{1}\left(  \frac{2\sqrt{t\left|
x\right|  }}{\beta\left(  \varepsilon\mu\right)  ^{\frac{1}{4}}}\right)
\right)  .
\end{align*}

Let us notice that $f$ fulfills the causality principle requirement which
guarantees that its convolution with the function from the right-hand side of
(\ref{Amax}) gives us the unique physically meaningful solution of the
inhomogeneous Maxwell system (\ref{Max1})-(\ref{Max3}) in a whole space.

\section{Inhomogeneous media}

Consider Maxwell's equations in a nonchiral inhomogeneous medium. Thus we
assume that $\varepsilon$ and $\mu$ are functions of coordinates:%
\[
\varepsilon=\varepsilon(x)\text{ \quad and\quad}\mu=\mu(x).
\]
Then the Maxwell system has the following form%
\begin{equation}
\operatorname{rot}\mathbf{H}=\varepsilon\partial_{t}\mathbf{E}+\mathbf{j,}
\label{Min1}%
\end{equation}%
\begin{equation}
\operatorname{rot}\mathbf{E}=-\mu\partial_{t}\mathbf{H}, \label{Min2}%
\end{equation}%
\begin{equation}
\operatorname{div}(\varepsilon\mathbf{E)}=\mathbf{\rho}, \label{Min3}%
\end{equation}%
\begin{equation}
\operatorname{div}\mathbf{(}\mu\mathbf{H)}=0. \label{Min4}%
\end{equation}
In this section following the procedure exposed in \cite{AQA} we show that
this system of equations can be written in the form of a single biquaternionic equation.

Equations (\ref{Min3}) and (\ref{Min4}) can be written as follows%
\[
\operatorname*{div}\mathbf{E}+<\frac{\operatorname{grad}\varepsilon
}{\varepsilon},\mathbf{E}>=\frac{\mathbf{\rho}}{\varepsilon}%
\]
and%
\[
\operatorname*{div}\mathbf{H}+<\frac{\operatorname{grad}\mu}{\mu}%
,\mathbf{H}>=0.
\]
Combining these equations with (\ref{Min1}) and (\ref{Min2}) we obtain the
Maxwell system in the form%
\begin{equation}
D\mathbf{E}=<\frac{\operatorname{grad}\varepsilon}{\varepsilon},\mathbf{E}%
>-\mu\partial_{t}\mathbf{H}-\frac{\mathbf{\rho}}{\varepsilon} \label{Min11}%
\end{equation}
and%
\begin{equation}
D\mathbf{H}=<\frac{\operatorname{grad}\mu}{\mu},\mathbf{H}>+\varepsilon
\partial_{t}\mathbf{E}+\mathbf{j}. \label{Min12}%
\end{equation}
Let us make a simple observation: the scalar product of two vectors
$\overrightarrow{p}$ and $\overrightarrow{q}$ can be written as follows%
\[
<\overrightarrow{p},\overrightarrow{q}>=-\frac{1}{2}(^{\overrightarrow{p}%
}M+M^{\overrightarrow{p}})\overrightarrow{q}.
\]
Using this fact, from (\ref{Min11}) and (\ref{Min12}) we obtain the pair of
equations%
\begin{equation}
(D+\frac{1}{2}\frac{\operatorname{grad}\varepsilon}{\varepsilon}%
)\mathbf{E}=-\frac{1}{2}M^{\frac{\operatorname{grad}\varepsilon}{\varepsilon}%
}\mathbf{E}-\mu\partial_{t}\mathbf{H}-\frac{\mathbf{\rho}}{\varepsilon}
\label{Min21}%
\end{equation}
and%
\begin{equation}
(D+\frac{1}{2}\frac{\operatorname{grad}\mu}{\mu})\mathbf{H}=-\frac{1}%
{2}M^{\frac{\operatorname{grad}\mu}{\mu}}\mathbf{H}+\varepsilon\partial
_{t}\mathbf{E}+\mathbf{j}. \label{Min22}%
\end{equation}
Note that
\[
\frac{1}{2}\frac{\operatorname{grad}\varepsilon}{\varepsilon}=\frac
{\operatorname{grad}\sqrt{\varepsilon}}{\sqrt{\varepsilon}}.
\]
Then (\ref{Min21}) can be rewritten in the following form%
\begin{equation}
\frac{1}{\sqrt{\varepsilon}}D(\sqrt{\varepsilon}\cdot\mathbf{E)}%
+\mathbf{E}\cdot\overrightarrow{\varepsilon}=-\mu\partial_{t}\mathbf{H}%
-\frac{\mathbf{\rho}}{\varepsilon}, \label{Min31}%
\end{equation}
where%
\[
\overrightarrow{\varepsilon}:=\frac{\operatorname{grad}\sqrt{\varepsilon}%
}{\sqrt{\varepsilon}}.
\]
Analogously, (\ref{Min22}) takes the form%
\begin{equation}
\frac{1}{\sqrt{\mu}}D(\sqrt{\mu}\cdot\mathbf{H)}+\mathbf{H}\cdot
\overrightarrow{\mu}=\varepsilon\partial_{t}\mathbf{E}+\mathbf{j,}
\label{Min32}%
\end{equation}
where%
\[
\overrightarrow{\mu}:=\frac{\operatorname{grad}\sqrt{\mu}}{\sqrt{\mu}}.
\]
Introducing the notations%
\[
\overrightarrow{\mathcal{E}}:=\sqrt{\varepsilon}\mathbf{E,\qquad
}\overrightarrow{\mathcal{H}}:=\sqrt{\mu}\mathbf{H,}%
\]
multiplying (\ref{Min31}) by $\sqrt{\varepsilon}$ and (\ref{Min32}) by
$\sqrt{\mu}$ we arrive at the equations%
\begin{equation}
(D+M^{\overrightarrow{\varepsilon}})\overrightarrow{\mathcal{E}}=-\frac{1}%
{c}\partial_{t}\overrightarrow{\mathcal{H}}-\frac{\mathbf{\rho}}%
{\sqrt{\varepsilon}}, \label{Minq1}%
\end{equation}
and%
\begin{equation}
(D+M^{\overrightarrow{\mu}})\overrightarrow{\mathcal{H}}=\frac{1}{c}%
\partial_{t}\overrightarrow{\mathcal{E}}+\sqrt{\mu}\mathbf{j}, \label{Minq2}%
\end{equation}
where as before $c=1/\sqrt{\varepsilon\mu}$ is the speed of propagation of
electromagnetic waves in the medium.

Equations (\ref{Minq1}) and (\ref{Minq2}) \ can be rewritten in an even more
elegant form. Consider the function%
\[
\mathbf{V}:=\overrightarrow{\mathcal{E}}+i\overrightarrow{\mathcal{H}}.
\]
Let us apply to it the biquaternionic operator
\[
\frac{1}{c}\partial_{t}+iD.
\]
We obtain%
\[
(\frac{1}{c}\partial_{t}+iD)\mathbf{V}=\frac{1}{c}\partial_{t}\overrightarrow
{\mathcal{E}}-D\overrightarrow{\mathcal{H}}+i(\frac{1}{c}\partial
_{t}\overrightarrow{\mathcal{H}}+D\overrightarrow{\mathcal{E}}).
\]
Applying (\ref{Minq2}) and (\ref{Minq1}) to the real and imaginary parts of
this equation gives\bigskip%
\begin{equation}
(\frac{1}{c}\partial_{t}+iD)\mathbf{V}=-i(M^{\overrightarrow{\varepsilon}%
}\overrightarrow{\mathcal{E}}+iM^{\overrightarrow{\mu}}\overrightarrow
{\mathcal{H}})-\sqrt{\mu}\mathbf{j}-\frac{i\mathbf{\rho}}{\sqrt{\varepsilon}}.
\label{vsp4111}%
\end{equation}
Note that
\[
\overrightarrow{\mathcal{E}}=\frac{1}{2}(\mathbf{V}+\mathbf{V}^{\ast}%
)\qquad\text{and\qquad}\overrightarrow{\mathcal{H}}=\frac{1}{2i}%
(\mathbf{V}-\mathbf{V}^{\ast}).
\]
Hence%
\[
M^{\overrightarrow{\varepsilon}}\overrightarrow{\mathcal{E}}%
+iM^{\overrightarrow{\mu}}\overrightarrow{\mathcal{H}}=\frac{1}{2}%
(M^{(\overrightarrow{\varepsilon}+\overrightarrow{\mu})}\mathbf{V+}%
M^{(\overrightarrow{\varepsilon}-\overrightarrow{\mu})}\mathbf{V}^{\ast}).
\]
Let us notice that
\[
\overrightarrow{\varepsilon}+\overrightarrow{\mu}=-\frac{\operatorname{grad}%
c}{c}\qquad\text{and\qquad}\overrightarrow{\varepsilon}-\overrightarrow{\mu
}=-\frac{\operatorname{grad}W}{W},
\]
where $W=\sqrt{\mu}/\sqrt{\varepsilon}$ is the intrinsic wave impedance of the
medium. Denote%
\begin{equation}
\overrightarrow{c}:=\frac{\operatorname{grad}\sqrt{c}}{\sqrt{c}}%
\qquad\text{and\qquad}\overrightarrow{W}:=\frac{\operatorname{grad}\sqrt{W}%
}{\sqrt{W}}. \label{definitionofcandw}%
\end{equation}
Then%
\[
M^{\overrightarrow{\varepsilon}}\overrightarrow{\mathcal{E}}%
+iM^{\overrightarrow{\mu}}\overrightarrow{\mathcal{H}}=-(M^{\overrightarrow
{c}}\mathbf{V+}M^{\overrightarrow{W}}\mathbf{V}^{\ast}).
\]
From (\ref{vsp4111}) we obtain the Maxwell equations for an inhomogeneous
medium in the following form%
\begin{equation}
(\frac{1}{c}\partial_{t}+iD)\mathbf{V-}M^{i\overrightarrow{c}}\mathbf{V-}%
M^{i\overrightarrow{W}}\mathbf{V}^{\ast}=-(\sqrt{\mu}\mathbf{j}+\frac
{i\mathbf{\rho}}{\sqrt{\varepsilon}}). \label{Maxmain}%
\end{equation}
This equation first obtained in \cite{KrISAAC}, \cite{AQA} is completely
equivalent to the Maxwell system (\ref{Min1})-(\ref{Min4}) and represents
Maxwell's equation for inhomogeneous media in a quaternionic form. We
formulate this as the following statement.

\begin{theorem}
Let $c=1/\sqrt{\varepsilon\mu}$, $\overrightarrow{c}$ and $\overrightarrow{W}$
be defined by (\ref{definitionofcandw}). Then two real-valued vectors
$\mathbf{E}$ and $\mathbf{H}$ are solutions of the system (\ref{Min1}%
)-(\ref{Min4}) iff the purely vectorial biquaternion $\mathbf{V=}%
\sqrt{\varepsilon}\mathbf{E}+i\sqrt{\mu}\mathbf{H}$ is a solution of
(\ref{Maxmain}).
\end{theorem}

Note that if $\varepsilon$ and $\mu$ are constant (a homogeneous medium),
equation (\ref{Maxmain}) turns into a well known (at least since the work of
C. Lanczos \cite{LanczosPhD}) biquaternionic reformulation of the Maxwell
system in a vacuum which was rediscovered by many researchers (e.g.,
\cite{Imaeda} and comments in \cite{GsponerHurni04}).

\begin{remark}
Equation (\ref{Maxmain}) can be considered as a generalization of the Vekua
equation, well known in complex analysis, that describes generalized analytic
functions \cite{Vekua}. Recently in \cite{Mal98} using the L. Bers approach
\cite{Berskniga} another generalization of the Vekua equation was considered.
Most likely some of the interesting results discussed in \cite{Mal98} can be
obtained for (\ref{Maxmain}) also. Their physical meaning would be of great interest.
\end{remark}

\subsection{Static case and factorization of the Schr\"{o}dinger
operator\label{Static}}

When the vectors of the electromagnetic field do not depend on time from
(\ref{Minq1}) and (\ref{Minq2}) we obtain two independent equations%
\[
(D+M^{\overrightarrow{\varepsilon}})\overrightarrow{\mathcal{E}}%
=-\frac{\mathbf{\rho}}{\sqrt{\varepsilon}}%
\]
and%
\[
(D+M^{\overrightarrow{\mu}})\overrightarrow{\mathcal{H}}=\sqrt{\mu}\mathbf{j.}%
\]

Let us consider the sourceless situation, that is we are interested in the
solutions for the operator $D+M^{\overrightarrow{\alpha}}$, where the complex
quaternion $\overrightarrow{\alpha}$ represents $\overrightarrow{\varepsilon}$
or $\overrightarrow{\mu}$ and has the form
\[
\overrightarrow{\alpha}=\frac{\operatorname{grad}f}{f}.
\]
The scalar function $f$ is different from zero.

Note that the study of the operator $D+^{\overrightarrow{\alpha}}M$
practically reduces to that of $D$, as shown in \cite{Sproessig}. In the case
of the operator $D+M^{\overrightarrow{\alpha}}$ (which can be called the
static Maxwell operator) the situation is quite different.

Consider the equation%

\begin{equation}
\left(  -\Delta+\nu\right)  g=0\qquad\text{in }\Omega\label{Schr3}%
\end{equation}
where $\nu$ and $g$ are complex valued functions, and $\Omega$ is a domain in
$\mathbb{R}^{3}$. We assume that $g$ is twice continuously differentiable.

\begin{theorem}
\label{Fact3} \cite{KrJPhys06} Let $f$ be a nonvanishing particular solution
of (\ref{Schr3}). Then for any scalar twice continuously differentiable
function $g$ the following equality holds,%
\begin{equation}
(D+M^{\frac{Df}{f}})(D-M^{\frac{Df}{f}})g=\left(  -\Delta+\nu\right)  g.
\label{FactSchr3}%
\end{equation}

\end{theorem}

\begin{remark}
The factorization (\ref{FactSchr3}) was obtained in \cite{Swansolo},
\cite{Swan} in a form which required a solution of an associated
biquaternionic Riccati equation. In \cite{KKW} it was shown that the solution
has necessarily the form $Df/f$ with $f$ being a solution of (\ref{Schr3}).
\end{remark}

\begin{remark}
Theorem \ref{Fact3} generalizes theorem 21 from \cite{KrRecentDevelopments}.
\end{remark}

\begin{remark}
\label{RemFromSchrToFirst}As $g$ in (\ref{FactSchr3}) is a scalar function,
the factorization of the Schr\"{o}dinger operator can be written in the
following form%
\[
(D+M^{\frac{Df}{f}})fD(f^{-1}g)=\left(  -\Delta+\nu\right)  g,
\]
from which it is obvious that if $g$ is a solution of (\ref{Schr3}) then the
vector $\mathbf{F}=fD(f^{-1}g)$ is a solution of the equation%
\begin{equation}
(D+M^{\frac{Df}{f}})\mathbf{F}=0\qquad\text{in }\Omega. \label{D+MDf}%
\end{equation}
The inverse result is given by the next statement where the following notation
is used%
\[
\mathcal{A}[\mathbf{G}](x,y,z)=%
%TCIMACRO{\dint \limits_{x_{0}}^{x}}%
%BeginExpansion
{\displaystyle\int\limits_{x_{0}}^{x}}
%EndExpansion
G_{1}(\xi,y_{0},z_{0})d\xi+%
%TCIMACRO{\dint \limits_{y_{0}}^{y}}%
%BeginExpansion
{\displaystyle\int\limits_{y_{0}}^{y}}
%EndExpansion
G_{2}(x,\zeta,z_{0})d\zeta+%
%TCIMACRO{\dint \limits_{z_{0}}^{z}}%
%BeginExpansion
{\displaystyle\int\limits_{z_{0}}^{z}}
%EndExpansion
G_{3}(x,y,\eta)d\eta+C
\]
($C$ is an arbitrary complex constant).
\end{remark}

\begin{theorem}
\label{ThFromDiracToSchr}\cite{KrJPhys06} Let $\mathbf{F}$ be a solution of
(\ref{D+MDf}) in a simply connected domain $\Omega$. Then $g=f\mathcal{A}%
[f^{-1}\mathbf{F}]$ is a solution of (\ref{Schr3}).
\end{theorem}

Moreover, a factorization of the operator $\operatorname{div}%
p\operatorname{grad}+q$ is valid also.

\begin{theorem}
\cite{KrJPhys06} Let $u_{0}$ be a nonvanishing particular solution of the
equation%
\begin{equation}
(\operatorname{div}p\operatorname{grad}+q)u=0\text{\qquad in }\Omega
\subset\mathbb{R}^{3} \label{maineq3}%
\end{equation}
with $p$, $q$ and $u$ being complex valued functions, $p\in C^{2}(\Omega)$ and
$p\neq0$ in $\Omega$. Then for any scalar function $\varphi\in C^{2}(\Omega)$
the following equality holds%
\begin{equation}
(\operatorname{div}p\operatorname{grad}+q)\varphi=-p^{1/2}(D+M^{\frac{Df}{f}%
})(D-M^{\frac{Df}{f}})p^{1/2}\varphi\label{mainfact3}%
\end{equation}
where $f=p^{1/2}u_{0}$.
\end{theorem}

Thus, if $u$ is a solution of equation (\ref{maineq3}) then
\[
\mathbf{F}=fD(f^{-1}p^{1/2}u)=fD(u_{0}^{-1}u)
\]
is a solution of equation (\ref{D+MDf}) (see remark \ref{RemFromSchrToFirst}).
The inverse result has the following form.

\begin{theorem}
\cite{KrJPhys06} Let $\mathbf{F}$ be a solution of equation (\ref{D+MDf}) in a
simply connected domain $\Omega,$ where $f=p^{1/2}u_{0}$ and $u_{0}$ be a
nonvanishing particular solution of (\ref{maineq3}). Then
\[
u=u_{0}\mathcal{A}[f^{-1}\mathbf{F}]
\]
is a solution of (\ref{maineq3}).
\end{theorem}

Notice that due to the fact that in (\ref{mainfact3}) $\varphi$ is scalar, we
can rewrite the equality in the form%
\[
(\operatorname{div}p\operatorname{grad}+q)\varphi=-p^{1/2}(D+M^{\frac{Df}{f}%
})(D-\frac{Df}{f}C_{H})p^{1/2}\varphi.
\]
Now, consider the equation
\begin{equation}
(D-\frac{Df}{f}C_{H})W=0, \label{Vekuamain3}%
\end{equation}
where $W$ is an $\mathbb{H}(\mathbb{C})$-valued function. Equation
(\ref{Vekuamain3}) is a direct generalization of the main Vekua equation
considered in \cite{KrRecentDevelopments}. Moreover, we show that it preserves
some of its important properties.

\begin{theorem}
\cite{KrJPhys06} Let $W=W_{0}+\mathbf{W}$ be a solution of (\ref{Vekuamain3}).
Then $W_{0}$ is a solution of the stationary Schr\"{o}dinger equation
\begin{equation}
-\Delta W_{0}+\nu W_{0}=0, \label{eqSchr3}%
\end{equation}
where $\nu=\Delta f/f$. Moreover, the function $u=f^{-1}W_{0}$ is a solution
of the equation
\begin{equation}
\operatorname{div}(f^{2}\operatorname{grad}u)=0 \label{eqscpart}%
\end{equation}
and the vector function $\mathbf{v}=f\mathbf{W}$ is a solution of the
equation
\begin{equation}
\operatorname{rot}(f^{-2}\operatorname{rot}\mathbf{v})=0. \label{eqvecpart}%
\end{equation}

\end{theorem}

\begin{remark}
\cite{KrJPhys06} Observe that the functions
\[
F_{0}=f,\quad F_{1}=\frac{i_{1}}{f},\quad F_{2}=\frac{i_{2}}{f},\quad
F_{3}=\frac{i_{3}}{f}%
\]
give us a generating quartet for the equation (\ref{Vekuamain3}). They are
solutions of (\ref{Vekuamain3}) and obviously any $\mathbb{H}(\mathbb{C}%
)$-valued function $W$ can be represented in the form%
\[
W=%
%TCIMACRO{\dsum \limits_{j=0}^{3}}%
%BeginExpansion
{\displaystyle\sum\limits_{j=0}^{3}}
%EndExpansion
\varphi_{j}F_{j},
\]
where $\varphi_{j}$ are complex valued functions. It is easy to verify that
the function $W$ is a solution of (\ref{Vekuamain3}) iff%
\begin{equation}%
%TCIMACRO{\dsum \limits_{j=0}^{3}}%
%BeginExpansion
{\displaystyle\sum\limits_{j=0}^{3}}
%EndExpansion
\left(  D\varphi_{j}\right)  F_{j}=0 \label{Vekuamain3second}%
\end{equation}
in a complete analogy with the two-dimensional case. Denote
\[
w=\varphi_{0}+\varphi_{1}i_{1}+\varphi_{2}i_{2}+\varphi_{3}i_{3}.
\]
Then (\ref{Vekuamain3second}) can be written as follows%
\[
D(w+\overline{w})f+D(w-\overline{w})\frac{1}{f}=0
\]
which is equivalent to the equation%
\[
Dw=\frac{1-f^{2}}{1+f^{2}}D\overline{w}.
\]

\end{remark}

\begin{remark}
The results of this section remain valid in the $n$-dimensional situation if
instead of quaternions the Clifford algebra $Cl_{0,n}$ (see, e.g., \cite{BDS},
\cite{GS2}) is considered. The operator $D$ is then introduced as follows
$D=\sum_{j=1}^{n}e_{j}\frac{\partial}{\partial x_{j}}$ where $e_{j}$ are the
basic basis elements of the Clifford algebra.
\end{remark}

\begin{acknowledgement}
The authors wish to express their gratitude to CONACYT for supporting this
work via the research project 50424.
\end{acknowledgement}

\end{document}